\documentclass{article}

\usepackage[total={6in, 8in}]{geometry}

\usepackage{booktabs} 
\usepackage[ruled]{algorithm2e} 

\SetAlFnt{\small}
\SetAlCapFnt{\small}
\SetAlCapNameFnt{\small}
\SetAlCapHSkip{0pt}
\IncMargin{-\parindent}
\usepackage{graphicx}

\usepackage{amssymb,amsthm,amsmath}
\usepackage{mathtools}
\usepackage{booktabs}
\usepackage{paralist}
\usepackage{tikz-cd}
\usepackage{tikz-network}
\usepackage{caption}
\usepackage{subcaption}
\usepackage{mathdots}
\usepackage{csquotes}
\usepackage{todonotes}
\usepackage{xargs}
\usepackage{comment}
\usepackage{standalone}
\usepackage{sgame}
\usepackage{bm}
\usepackage{natbib}
\usepackage{url}

\usepackage{subcaption}

\tikzset{ball/.style={circle, draw, fill=black,inner sep=0pt, minimum width=4pt}}

\usepackage{pgfplots}
\pgfplotsset{compat = newest}

\usepackage{xcolor}
\usetikzlibrary{decorations.markings}

\usepackage{tikz}
\usetikzlibrary{arrows.meta}
\tikzset{label/.style = {inner sep=1pt, fill=white}}
\tikzset{nd/.style={inner sep=1pt}}
\tikzset{>=Latex}
\tikzset{arc/.style = {->, semithick, >=Latex}}

\theoremstyle{plain}
\newtheorem{thm}{Theorem}[section]

\theoremstyle{definition}
\newtheorem{defn}[thm]{Definition}

\newtheorem{lem}[thm]{Lemma}
\newtheorem{corol}[thm]{Corollary}

\newtheorem{prop}[thm]{Proposition}

\newcommand{\real}{\mathbb{R}}

\newcommand{\exptu}{\mathbb{U}}

\newcommand{\arc}[3][]{ #2 \xlongrightarrow{#1} #3}

\DeclareMathOperator{\intr}{int}

\newcommand{\dx}{\mathrm{d}}

\DeclareMathOperator{\rd}{rd}

\DeclareMathOperator{\softmax}{softmax}
\DeclareMathOperator{\BRD}{BRD}
\DeclareMathOperator{\brd}{brd}

\title{Computing stable limit cycles of learning in games}
\date{}

\author{Oliver Biggar and Christos Papadimitriou}

\begin{document}

\maketitle

\begin{abstract}
Many well-studied learning dynamics, such as fictitious play and the replicator, are known to not converge in general $N$-player games. The simplest mode of non-convergence is cyclical or periodic behavior. Such cycles are fundamental objects, and have inspired a number of significant insights in the field, beginning with the pioneering work of Shapley (1964). However a central question remains unanswered: which cycles are stable under game dynamics? In this paper we give a complete and computational answer to this question for the two best-studied dynamics, fictitious play/best-response dynamics and the replicator dynamic. We show (1) that a periodic sequence of profiles is stable under one of these dynamics if and only it is stable under the other, and (2) we provide a polynomial-time spectral stability test to determine whether a given periodic sequence is stable under either dynamic. Finally, we give an entirely `structural' sufficient condition for stability: every cycle that is a sink equilibrium of the preference graph of the game is stable, and moreover it is an attractor of the replicator dynamic. This result generalizes the famous theorems of Shapley (1964) and Jordan (1993), and extends the frontier of recent work relating the preference graph to the replicator attractors.
\end{abstract}

\section{Introduction} \label{sec:intro}

In 1964, Lloyd Shapley proved a now-famous result about learning in games: \emph{fictitious play}---a natural learning algorithm where each player selects a best response to the past play of the other players---does not converge to Nash equilibria in every game \citep{shapley_topics_1964}. This finding upended the widely-held belief that learning was the natural justification for Nash equilibrium play, and was the beginning of a research tradition on the interaction betwen learning and Nash equilibria that now encompasses the field of \emph{game dynamics}. 

Shapley's counterexample is very simple. He presented a symmetric, $3\times 3$ game (Figure~\ref{fig:shapley})---not dissimilar to Rock-Paper-Scissors (RPS)---possessing a unique Nash equilibrium (NE). Unlike RPS, this equilibrium is \emph{unstable} under fictitious play (FP)---nearby trajectories spiral away from it, eventually converging to a 6-cycle, now called the \emph{Shapley polygon} \cite{gaunersdorfer_fictitious_1995,benaim_learning_2009}. 

Over time, a growing body of work found that the convergence to the cycle in Shapley's game is a surprisingly robust phenomenon. First, it generalizes to many other dynamics \cite{gaunersdorfer_fictitious_1995,benaim_perturbations_2012,viossat_no-regret_2013}, notably including the \emph{multiplicative weights update} \citep{arora2012multiplicative} learning algorithm and its continuous-time analog, the \emph{replicator dynamic (RD)} \cite{taylor_evolutionary_1978,smith1973logic}. Unlike FP and its generalized form, the \emph{best-response dynamics} (BRD), RD is smooth and arises from first principles in both learning---as a regularized form of FP---and \emph{evolutionary game theory}
\cite{sandholm2010population}. Second, the phenomenon extends beyond this one game, with various papers in economics and computer science observing cyclic behavior in game dynamics, both theoretically \cite{zeeman_population_1980, krishna_convergence_1998, foster_nonconvergence_1998, hofbauer1998evolutionary,imhof2005evolutionary,kleinberg_beyond_2011,mertikopoulos_cycles_2018,boone_darwin_2019} and empirically \cite{mckane2005predator,cason2010testing,xu2013cycle,cason2014cycles,pangallo_best_2019} (see Section~\ref{sec: related}). The other particularly notable non-convergent game example came from Jordan \cite{jordan_three_1993}, who constructed a $2\times 2\times 2$ game where, like Shapley's, FP converges to a 6-sided cycle instead of the unique Nash equilibrium (see Figure~\ref{fig:jordan}). Later, Kleinberg et al. \cite{kleinberg_beyond_2011} showed that for a parametrized family of utilities in Jordan's game, the RD also converges to a 6-cycle, and further that the payoff outcomes on this cycle can be \emph{much better} than the equilibrium outcome. 

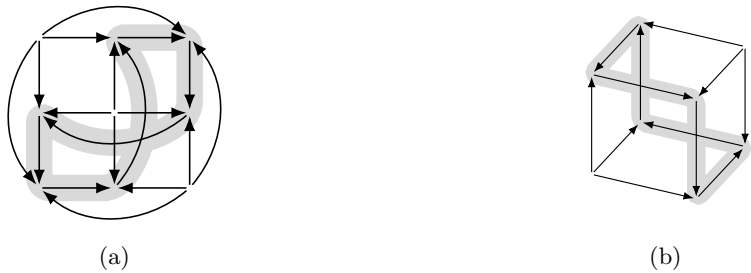
\begin{figure}
    \centering
    \begin{subfigure}{.45\textwidth}
        \centering
        \begin{tikzpicture}

\draw [rounded corners, color = gray!30, line width = 10pt] (2,2) to (2,1) to[out = 220, in = -40] (0,1) to (0,0) to (1,0) to[in = -50, out = 50] (1,2) -- cycle;

    \node[nd] (1) at (0,0) {};
    \node[nd] (2) at (1,0) {};
    \node[nd] (3) at (2,0) {};
    \node[nd] (4) at (0,1) {};
    \node[nd] (5) at (1,1) {};
    \node[nd] (6) at (2,1) {};
    \node[nd] (7) at (0,2) {};
    \node[nd] (8) at (1,2) {};
    \node[nd] (9) at (2,2) {};
    
    \draw[arc] (4) to (1);
    \draw[arc] (3) to[in = -40, out = 220] (1);
    \draw[arc] (3) to[out = 50, in = -50] (9);
    \draw[arc] (8) to (9);
    \draw[arc] (5) to (8);
    \draw[arc] (5) to (4);
    
    \draw[arc] (7) to (4);
    \draw[arc] (7) to[in = 130, out = -130] (1);
    \draw[arc] (7) to (8);
    \draw[arc] (7) to[in = 140, out = 40] (9);
    \draw[arc] (1) to (2);
    \draw[arc] (3) to (2);
    \draw[arc] (5) to (2);
    \draw[arc] (2) to[in = -50, out = 50] (8);
    \draw[arc] (3) to (6);
    \draw[arc] (9) to (6);
    \draw[arc] (5) to (6);
    \draw[arc] (6) to[out = 220, in = -40] (4);

 \end{tikzpicture}
        \caption{}
        \label{fig:shapley}
    \end{subfigure}
    \quad
    \begin{subfigure}{.45\textwidth}
        \centering
        \begin{tikzpicture}[scale=.7]
\begin{axis}[
    axis line style = {opacity=0},
    xmin=-0.1, xmax=1.1,
    ymin=-0.1, ymax=1.1,
    zmin=-0.1, zmax=1.1,
    axis equal image,
    xticklabels = {},
    yticklabels = {},
    zticklabels = {},
    xtick style = {draw=none},
    ytick style = {draw=none},
    ztick style = {draw=none},
]

    \draw [rounded corners, color = gray!30, line width = 10pt] (1,0,1) to (0,0,1) to (0,1,1) to (0,1,0) to (1,1,0) to (1,0,0) -- cycle;

    
    \node[nd] (hhh) at (0,0,0) {};
    \node[nd] (thh) at (1,0,0) {};
    \node[nd] (hht) at (0,0,1) {};
    \node[nd] (tht) at (1,0,1) {};
    
    \node[nd] (hth) at (0,1,0) {};
    \node[nd] (tth) at (1,1,0) {};
    \node[nd] (htt) at (0,1,1) {};
    \node[nd] (ttt) at (1,1,1) {};
    
    
    \draw[arc] (hhh) to (hht);
    \draw[arc] (hhh) to (hth);
    \draw[arc] (hhh) to (thh);
    \draw[arc] (ttt) to (htt);
    \draw[arc] (ttt) to (tth);
    \draw[arc] (ttt) to (tht);
    \draw[arc] (htt) to (hht);
    \draw[arc] (hht) to (tht);
    \draw[arc] (hth) to (htt);
    \draw[arc] (tth) to (hth);
    \draw[arc] (thh) to (tth);
    \draw[arc] (tht) to (thh);
\end{axis}
 \end{tikzpicture}
        \caption{}
        \label{fig:jordan}
    \end{subfigure}
    \caption{The preference graphs \cite{biggar_preference_2025} of the famous non-convergent games of Shapley \cite{shapley_topics_1964} and Jordan \cite{jordan_three_1993}. In these games, fictitious play and the replicator dynamic do not converge to Nash equilibria, and instead converge to the highlighted 6-cycle, which is also the unique \emph{sink equilibrium} \cite{goemans_sink_2005}. We show in this paper that cycles like these, without outgoing arcs, are always stable under best-response dynamics and the replicator dynamic.}
    \label{fig:shapley and jordan}
\end{figure}

These observations---that non-convergence to Nash equilibria, and the presence of cycles, are general game-theoretic phenomena---motivated a shift in thinking. Instead, Papadimitriou and Piliouras \cite{papadimitriou_game_2019} proposed that we should view the `meaning' of the game as the stable outcomes of game dynamics, whether they are Nash equilibria or, more often, some complex pattern.

For many game dynamics (including BRD and the RD), it is easy to see that, if the dynamic converges to a single point from any fully-mixed initial strategy profile, then that point must be a Nash equilibrium \cite{hofbauer1998evolutionary}. The next more complex outcome occurs when the dynamic converges to a 
\emph{cycle}, that is, a periodic path, whose nodes are pure strategy profiles---as in Shapley's and Jordan's games. Such cycles form an intermediate case between Nash convergence and the opposite extreme of chaotic, totally unstructured orbits and so represent a complex, non-static behavior that remains interpretable \cite{sparrow2008fictitious}. They are a critical first waypoint towards understanding more complex limit sets. As such, they represent one of the best-researched classes of non-convergent behavior in game dynamics, whose major studies include \cite{shapley_topics_1964,jordan_three_1993,gaunersdorfer_fictitious_1995,krishna_convergence_1998,foster_nonconvergence_1998,hofbauer_evolutionary_2003,hofbauer_time_2009,benaim_learning_2009,kleinberg_beyond_2011} (see Section~\ref{sec: related}). Despite this progress, however, an understanding of stable cycles in general $N$-player games remains elusive. That is our goal in this paper: \emph{to identify and compute the stable limit cycles of BRD and the RD.}

\subsection*{Our contributions}


We begin by identifying a fundamental computational problem for game dynamics: \emph{Given a game, which cycles of outcomes can occur stably under BRD?  Under RD?} We give an answer to this question through two theorems: first, we show that for BRD (and hence for FP), stable periodic paths correspond to walks in the preference graph which satisfy a  spectral condition. This gives us a polynomial-time algorithm for verifying stability of periodic behavior under BRD, placing the problem of finding stable cycles into NP. Our second major contribution extends this to RD: we show that a cycle is stable under BRD if and only if it is stable under RD, immediately establishing that the same computational test works for RD, while at the same demonstrating the `robustness' of the stability property to the specific dynamic under consideration. This extends a long line of research connecting these two dynamics \cite{gaunersdorfer_fictitious_1995,benaim_perturbations_2012,viossat_no-regret_2013}, and allows us to leverage the `inherently combinatorial' structure of BRD to prove results about RD. 



One of the key properties of Shapley's original non-convergence proof was its combinatorial nature \cite{shapley_topics_1964}. Specifically, he defined his eponymous game by a collection of inequalities among utilities, rather than by a specific numerical instance. In today's terminology, these inequalities exactly define a graph known as the \emph{preference graph} \cite{biggar_graph_2023,biggar_preference_2025}, an object which captures the discrete structure of the utilities in a game. Papadimitriou and Piliouras hypothesized that the sink strongly connected components of the game (the \emph{sink equilibria} \cite{goemans_sink_2005}) capture key properties of the limit sets of the replicator dynamic. In a sequence of follow-up results, Biggar and Shames \cite{biggar_replicator_2023,biggar_attractor_2024} proved that the attractors of the replicator dynamic indeed closely correspond to the structure of the preference graph, with the sink equilibria completely determining the attractors in several classes of games \cite{biggar2025sink}.

The power of `graph-based' results like these---when they exist---is that they are robust to small changes in the utility functions, and provide intuition for how the \emph{structure} of the game affects the dynamic outcome. In this context, Shapley's game can be characterized by a simple cycle property: \emph{it is the smallest two-player preference graph whose sink equilibrium is a cycle of length longer than four}\footnote{Cycles of length 4 are the smallest possible, and form a special case where convergence to Nash equilibria takes place. See Section~\ref{sec: 4cycle} and \cite{krishna_convergence_1998}.} \cite{biggar_preference_2025}. Jordan's game too, when viewed as a graph, can be characterized by a similar property: \emph{it is the smallest $N$-player game whose sink equilibrium is a cycle of length longer than four} \cite{biggar_preference_2025}. Our third major contribution of this paper is a theorem in graph-theoretic line, which generalizes Shapley's and Jordan's findings: \emph{whenever a cycle is also a sink equilibrium, it is stable under BRD, and furthermore it is an attractor of the replicator dynamic}\footnote{Most stable cycles, while stable in the sense of attracting an open set of fully-mixed initial conditions, are not \emph{asymptotically stable} under the replicator, and hence are not attractors. This is because asymptotic stability requires uniform convergence on the boundary as well as the interior of the game. When the cycle is a sink equilibrium, this stronger condition holds.}. This provides the first proof that the stability of the 6-cycle in Jordan's game is completely determined by its graph structure\footnote{This is notable because proofs of convergence in Jordan's game are difficult, and typically rely on assumptions like 0/1-payoffs or symmetry \cite{jordan_three_1993,gaunersdorfer_fictitious_1995,benaim_perturbations_2012}. Kleinberg et al. \cite{kleinberg_beyond_2011} is the most general such convergence result we know of in this genre, and it only suffices to prove convergence for a one-parameter family of Jordan-graph games, and their proof is very long and technical.}. Our result also extends a recent theorem of Biggar and Papadimitriou \cite{biggar2025sink}, who showed that sink equilibria which are \emph{uniform-weighted cycles} (those where the weights on each arc on the cycle are equal) are always attractors for RD.



\subsection{Related work} \label{sec: related}


The fact that dynamics can exhibit cycling behavior originates in Shapley's \cite{shapley_topics_1964} paper, and this seeded a vast array of literature on non-convergence in games. Jordan's $2\times 2\times 2$ \cite{jordan_three_1993} gave an even simpler presentation of the non-convergence of BRD. Gaunersdorfer and Hofbauer \cite{gaunersdorfer1992time,gaunersdorfer_fictitious_1995} extended these ideas to the replicator dynamic, showing that the time-average of the replicator dynamic also converges in Shapley's and Jordan's games, and to the same polygonal orbit, which they dub the \emph{Shapley polygon}. This comparison was the first in a long and mathematically rich analysis of the relationship between BRD and various continuous perturbations of it, like the replicator \cite{benaim_learning_2009,hofbauer_time_2009,benaim_perturbations_2012,viossat_no-regret_2013}. These results in this line prove that no-regret dynamics that are perturbed best-response must converge to set that are invariant and chain transitive under BRD. This can be viewed as partial converse to our results, where we show that a specific stability criterion suffices to show convergence of both the replicator and BRD.

Monderer and Sela \cite{monderer_fictitious_1997} use the example of Shapley's to develop sufficient conditions on when cycling cannot occur under FP, giving a proof that CFP converges to Nash equilibria in $2\times 3$ games. This idea was generalized by Berger \cite{berger_fictitious_2005}, who showed that FP converges to Nash in $2\times n$ games. Interestingly, Berger's results can be interpreted in a `cycle' sense: he shows that almost every FP trajectory in a $2\times n$ game eventually reaches either a PNE or a \emph{4-cycle}, demonstrating that even in this class of games, the last-iterate behavior is often cyclic. The remainder of the statement follows from a special property of 4-cycles, that we discuss in Section~\ref{sec: 4cycle}: if FP follows a stable 4-cycle, the time-average converges to a Nash equilibrium. Four-cycles are special in this regard---Krishna and Sj\"ostrom \cite{krishna_convergence_1998} prove that if CFP follows a periodic sequence of length greater than four in a 2-player game with an interior Nash equilibrium, the time-average does not converge. The mathematical ideas in \cite{krishna_convergence_1998} lay the groundwork for our work, including developing the main ideas behind the Poincar\'e matrix we build in Section~\ref{sec: FP}, and the resultant eigenvalue analysis. Our Theorem~\ref{thm: stability test FP} implies that their result also holds in general $N$-player games. Similarly, the core ideas of the Poincar\'e map and their application to heteroclinic replicator cycles is thoroughly developed in the influential textbook of Hofbauer and Sigmund \cite{hofbauer1998evolutionary}. The particularly focus on `monocyclic games', which are symmetric/one-population games with one `spanning' best-response cycle.

In their work, Gaunersdorfer and Hofbauer \cite{gaunersdorfer_fictitious_1995} reference Gilboa and Matsui \cite{gilboa1991social} in noting that, in Shapley's game, ``the Shapley polygon is a much better candidate for the `solution' of a game than the unstable Nash equilibrium." This observation represents a shift in thinking from viewing cycling as aberrant non-convergent behavior to viewing it as a one of the fundamental possibilities for learning in games. The Shapley polygon inspired Benaim et al. \cite{benaim_learning_2009} to propose a new solution concept for monocyclic games: the TASP (time-average of the Shapley polygon). The TASP is motivated as a step towards a general dynamical solution concept, and the authors show that when using recency-biased CFP the players converge to a single `median point' of the cycle. They argue that the TASP is a better predictor than Nash not only theoretically, but also empirically, an idea explored in several works in experimental economics \cite{cason2010testing,xu2013cycle,cason2014cycles}. In this paper we describe exactly when stable Shapley polygons exist under BRD, and our sufficient condition in Theorem~\ref{thm: sink equilibria are stable} can be viewed as an $N$-player generalization of the monocyclicity condition.

This shift away from Nash-based concepts to dynamical ones in economics literature paralleled a similar shift in the computer science literature. In that context, the idea originated in the study of the Price of Anarchy \cite{roughgarden2005selfish,koutsoupias1999worst}. Goemans et al. \cite{goemans_sink_2005} observed that the algorithmic predictions of the Nash equilibrium-based Price of Anarchy performed poorly in dynamic settings, reflecting both the general non-convergence of dynamics \cite{milionis_impossibility_2023,hart_uncoupled_2003} and the computational difficulty of computing a Nash equilibrium \cite{daskalakis_complexity_2009,chen2009settling}. They suggested an alternative metric grounded in dynamics, called the Price of Sinking, based on a different solution concept called the \emph{sink equilibria}, which are the sink connected components of the best-response graph of the game. In a famous later work, Kleinberg et al. \cite{kleinberg_beyond_2011} showed that, in a parameterized family of $2\times 2\times 2$ games, the Nash equilibrium gave a far worse prediction of the behavior of the replicator dynamic than the sink equilibrium, which in this game was a 6-cycle, and further showed that the utility on the 6-cycle was arbitrarily higher than the Nash. This game turns out to be exactly Jordan's game\footnote{It appears that this was an independent discovery of this game, as the authors do not cite Jordan's paper.}, and this sink equilibrium defines the heteroclinic cycle corresponding to the Shapley polygon.

These ideas led to the proposal of \cite{papadimitriou_game_2019} to view the outcomes of dynamics as the meaning of the game, which we discussed in the introduction, as well as the suggestion that the preference graph and its sink equilibria might form a useful tool for this analysis \cite{biggar_graph_2023}. A recent survey of the preference graph properties is \cite{biggar_preference_2025}. This approach to game dynamics has led to several recent results \cite{biggar_replicator_2023,biggar_attractor_2024,biggar2025sink,hakim2024swim} demonstrating that the structure of the preference graph is indeed closely related to the attractors of the replicator dynamic. Most close to our work is the recent theorem of \cite{biggar2025sink}, who proved that sink equilibria which are uniform-weighted cycles obey a property called \emph{pseudoconvexity}, which guarantees that they are attractors. However, their result does not generalize to all cycles because cycles without uniform weights are not pseudoconvex. Out Theorem~\ref{thm: sink equilibria are stable} completes this picture, by proving that all cycles, regardless of weights, are attractors of the replicator dynamic.

\subsection*{Preliminaries} \label{sec: prelims}

Our paper studies $N$-player normal-form games, with finitely many strategies $S_1,S_2,\dots,S_N$ for each player \citep{myerson1997game}. The payoffs are defined by a utility function $u : \prod_{i=1}^{N}S_i \to \real^N$. 
A \emph{strategy profile} $p$ is an assignment of strategies to all player and we use $p_{-i}$ to denote an assignment of strategies to all players other than $i$. 
A \emph{subgame} is the game resulting from restricting each player to a subset of their strategies.
A \emph{mixed strategy} is a distribution over a player's pure strategies, and a \emph{mixed profile} is an assignment of a mixed strategy to each player. We also refer to strategies as \emph{pure strategies} and profiles as \emph{pure profiles} to distinguish them the mixed cases. We typically write a mixed profile as $x$ and denote $x^i$ for the distribution over player $i$'s strategies, and $x^i_s$ for the $s$-entry of player $i$'s distribution, where $s\in S_i$. The \emph{strategy space} of the game is the set of all mixed profiles $X=\prod_{i=1}^N \Delta_{|S_i|}$ where $\Delta_{|S_i|}$ are the simplices in $\real^{|S_i|}$. 
Any mixed profile $x$ naturally defines a product distribution over profiles, which we denote by 
\begin{equation} \label{product}
    z = x^1 \otimes x^2 \otimes \dots \otimes z^N.
\end{equation} That is, if $p = (s_1,s_2,\dots,s_N)$ is a strategy profile, and $x$ is a mixed strategy, then $z_p = x^1_{s_1}x^2_{s_2}\dots x^N_{s_N}$ is the $p$-entry of the distribution $z$. $z_{p_i}$ is defined similarly.


The definition of the utility function extends naturally to mixed profiles, via expectation. We denote this by $\exptu : X \to \real^N$, defined by
\begin{equation} \label{def: expected utility}
\exptu(x) = \sum_{s_1\in S_1} \sum_{s_2\in S_2} 
    \dots\sum_{s_N\in S_N} \left(\prod_{j = 1}^N x^j_{s_j}\right) u(s_1,s_2,\dots,s_N) = \sum_{p} z_p u(p)
\end{equation}



An important tool in our paper is the \emph{preference graph} of a game \citep{biggar_graph_2023,biggar_preference_2025}. This a directed graph 
whose nodes are the pure profiles of the game, and where there is a directed edge from $p$ to $q$ if (a) $p$ and $q$ differ in the strategy of one player, and (b) this player's utility in $p$ is greater than or equal to that in $q$. The arcs of the preference graph can be canonically given weights, where an arc is weighted by the difference between the two utilities. 
The {\em sink  equilibria} of a game are the sink strongly connected components of the preference graph of the game. For a recent summary of preference graph results, see \cite{biggar_preference_2025}.



\subsubsection*{Dynamical systems}


We study two game dynamics in this paper: the \emph{replicator dynamic} (RD) and the \emph{best-response dynamics} (BRD). First, the \emph{replicator dynamic}~\citep{sandholm2010population,hofbauer_evolutionary_2003} is flow on the strategy space of a game, defined as the solution of the following ordinary differential equation:  
\[
\dot x_s^i = x_s^i\left(\exptu_i(s; x_{-i}) - \sum_{r\in S_i} x_r^i \exptu_i(r; x_{-i}) \right)
,\]
where $x$ is a mixed profile, $i$ a player and $s$ a strategy. The second is the best-response dynamics (BRD) \cite{hofbauer_evolutionary_2003}, defined by 
\[
\dot x(t) \in BR(x(t)) - x(t)\ .
\]
Here, $BR(x)$ is the best-response function---given a mixed strategy profile $x$, it $BR_i(x)$ returns the set of best-responses to $x$ for player $i$, equivalently defined as $\arg\max u_i(x)$. We will somewhat abuse notation by writing $BR(x)$ to mean the elementwise Cartesian product $\prod_i BR(x^i)$.

In much the same way that the RD is the continuous-time generalization of the multiplicative weights algorithm \cite{arora2012multiplicative}, BRD is a continuous generalization of classical fictitious play algorithm from its original discrete-time setting. An alternative generalization is \emph{Continuous Fictitious Play} (CFP) \citep{viossat_no-regret_2013}, which is very similar to BRD except differs by a rescaling of velocity, giving: 
\[
\dot x(t) \in \frac{1}{t}(BR(x(t)) - x(t))
\]
This rescaling does not affect the limits, and we will work with BRD for its relatively simpler presentation. Note that unlike the replicator dynamics, this is a \emph{differential inclusion} \cite{aubin2008differential} rather than a differential equation. This is because at some points the best-response correspondence is multi-valued. However, at least one trajectory exists from every starting point, and indeed from almost all points exactly one trajectory exists \cite{gaunersdorfer_fictitious_1995,hofbauer_evolutionary_2003,viossat_no-regret_2013}. In fact, from almost all starting points the best-response function only becomes multi-valued for intervals of zero time (which do not affect the trajectory). We can focus our attention on these generic BRD paths, because our goal in this paper is to characterize periodic behavior that is \emph{persistent} (Definition~\ref{def: periodic play}), in the sense that it is followed from an open set of initial conditions, as in \cite{krishna_convergence_1998}.

All of these dynamics (and more generally, all of the class of Follow-The-Regularized-Leader (FTRL) dynamics \cite{vlatakis-gkaragkounis_no-regret_2020}), can be equivalently expressed in the space of `accumulated payoffs', which we here call simply \emph{payoff space}. This is a well-known transformation \cite{hofbauer_time_2009}, which leads to many useful insights, including the fact of \emph{volume conservation} \cite{akin_evolutionary_1984,hofbauer1998evolutionary,swenson_best-response_2018}, which we will make significant use of in this paper. In contrast to \emph{strategy space}, where the state is expressed in terms of mixed strategy profiles, payoff space expressed the state in terms of the total payoff accrued by each strategy for each player. Throughout this paper, we typically denote elements of strategy space by $x$ and payoff space by $w$, and we wil mostly work in payoff space. To see why the transformation to payoff space can be useful, observe that in payoff space RD can be equivalently expressed as the \emph{exponential weight} dynamic:
\begin{defn}[RD in payoff space] \label{def: rd payoff}
\[
    \dot w^i_s = \exptu^i_s(\softmax_{j\neq i} w)
\]
\end{defn}
where $\softmax_{j\neq i} w$ is applied elementwise on each player $j\neq i$, giving a mixed profile $x_{-i}$. One can recover the replicator equation by defining $x^i := \softmax w^i$ and taking the time derivative of $x^i$. BRD can be defined very similarly, in a way which reflects the similarities between these dynamics:
\begin{defn}[BRD in payoff space] \label{def: brd payoff}
    \[
\dot w^i_s \in \exptu^i_s(\arg\max_{j\neq i} (w))
\]
\end{defn}

For BRD, the dynamics from a generic starting point in payoff space consist of piecewise linear paths, where the best-response is a pure profile. This naturally gives us a sequence of pure profiles that we call the \emph{sequence of play}. The following definition formalizes this in a way such that this discretization can be applied to RD as well.

\begin{defn}[Switches and sequences of play]
    Let $w$ be a point in payoff space, and let $\phi$ be our dynamic. Given any point $\phi(w,t)$, the maximum $\arg\max \phi(w,t)$ defines a subgame of the game, and this sequence of subgames is a locally constant map in $t$. That is we get a sequence of times $t_i$ such that during each segment $t\in (t_i,t_{i+1})$, $\arg\max \phi(w,t) =: S_i$ is constant. We call the sequence of subgames $S_i$ (played for positive time intervals) the \emph{sequence of play}, and we call the times $t_i$ the \emph{switching times}.
\end{defn}



\subsection{Periodic and eventually periodic behavior}

Our goal in this paper is to characterize the possible {\em stable} periodic behaviors that can occur under game dynamics; by {\em stable behaviors} we mean behavior that is unchanged by small perturbations of the input. Cyclic or periodic behavior intuitively means that the same sequence of outcomes repeats indefinitely. In this section we give a precise formalization of this intuition.

For motivation, recall that under BRD, the gradient is locally constant, and so the trajectories take the form of piecewise linear paths. In each segment of this trajectory, all players play a best-response against the empirical distribution of their opponents, which defines the BRD sequence of play.
Hence we can naturally define \emph{periodic play} as a trajectory where the sequence of play is periodic, in that it consists of a repeated sequence of $K$ profiles for some $K$. However, we want also that this periodic play can occur robustly, in the sense that it is immune to small perturbations in the initial conditions.

\begin{defn}[Stable periodic play] \label{def: periodic play}
    Let $\phi$ be our dynamic (RD or BRD) and let $w$ be a point in payoff space. Given a periodic sequence of profiles $c$, we say the play from $w$ is \emph{periodic} if the sequence of play is periodic. We say the play sequence is \emph{persistent} if there is an open set $U$ of points in payoff space such that the sequence of play follows that sequence. We say that the play is \emph{stable} if it is both persistent and \emph{isolated}, in that for all $w\in U$ their normalized trajectories $w/\|w\|$ has the same limit set.
\end{defn}

The difference between persistent and stable is the additional criterion that the behavior be `isolated'. This technical-sounding requirement exists to distinguish the case when periodic paths are \emph{Lyapunov stable} but not \emph{asymptotically stable}, essentially because they lie on a manifold containing infinitely many different periodic paths. Being isolated guarantees that the limit behavior is unique, paths that are both persistent and isolated are robust to small perturbations of their input making them a generic source of periodic behavior in game dynamics. The special case of 4-cycles, which we discuss in Section~\ref{sec: 4cycle} later, demonstrates the difference between persistence and stability.

There is another important point to note here. If $w$ lies on a periodic orbit (after some finite time $T$, the dynamic maps $w$ back to $w$), then the sequence of play it generates is also periodic. However, 
for periodic play we do not require that the same amount of \emph{time} is spent playing each profile. Shapley's game provides the motivating example \cite{shapley_topics_1964}: while the CFP path follows the 6-cycle forever, the time spent at each profile increases geometrically with each lap of the cycle\footnote{It is this exponential increase that ensures that the trajectory does not converge to the Nash equilibrium.}. 
It is also for this reason that we talk about the limit set for the normalized sequence $w/\|w\|$---the unnormalized sequence may not have any limit points, because it goes to infinity.

This distinction is also important for defining periodic behavior for the replicator dynamic. In Shapley's game, the 6-cycle is not periodic under the replicator; in the neighborhood of the cycle, we converge to a {\em heteroclinic} 6-cycle---meaning that the dynamic moves around the cycle, stopping and restarting at the nodes. Intuitively, this play is periodic because the sequence of observed outcomes follows a periodic path between the six profiles on this heteroclinic cycle. This fits the definition of periodic play for the replicator we have here, in terms of the sequence of observed outcomes.

\section{Stable periodic play under best-response dynamics} \label{sec: FP}

In this section we characterize the possible stable periodic sequences under BRD. The key idea is to obtain a combinatorial representation of periodic paths under BRD, by identifying periodic walks in the preference graph of the game.

\subsection{Best-response dynamics and the preference graph} 

The piecewise-linear nature of generic BRD paths is closely connected to the structure of the preference graph, as the example of Shapley's game (Figure~\ref{fig:shapley}) demonstrates. The intuition for this connection is simple: each BRD switch between two pure profiles involves a unilateral deviation of a player to a better-response strategy, given the current empirical distribution of their opponent. These deviation define the arcs in the preference graph, so we would expect that generic switches must occur at arcs. This is indeed the case, as the following lemma shows.

\begin{lem} 
    If there is a BRD switch from a profile $p_1$ to $p_2$, then there is an arc in the preference graph from $p_1$ to $p_2$.
\end{lem}

The proof is elementary, and forms exist (with slightly different notation) in \cite{krishna_convergence_1998} and \cite{berger_two_2007}. Hence, a periodic sequence of play consists of a repeated sequence of BRD switches, which by this lemma must have arcs between them in the preference graph. The following is then immediate:

\begin{corol}
    Every periodic sequence of play under BRD follows a periodic walk in the preference graph.
\end{corol}

Note that we allow profiles to be used multiple times within a period of a periodic sequence of play. This means that the sequence may only be a \emph{walk} and not a \emph{cycle},  
and the length of the repeated sequence, while finite, can be potentially unbounded in the size of the game. Our method for testing periodic sequences for stability will scale with the length of the sequence we are testing. Now this walk
is the key idea behind our results, because it connects periodic play with a combinatorial object. In the next section, we show that this object is unique---if a periodic walk $c$ is generated by a stable periodic BRD sequence, then that sequence is unique for $c$.  This allows us to find stable periodic play by checking walks in the preference graph for stability.



\subsection{The Poincar\'e matrix of a periodic walk}


In dynamical systems, \emph{Poincar\'e maps} provide a powerful method for determining the stability of a periodic orbit. A Poincar\'e map is defined by selecting a hyperplane through the path of a given orbit, and computing the image of a given point on the hyperplane near the orbit at the time when it next intersects that hyperplane. This constructs a discrete-time dynamical system on this hyperplane where the periodic orbit in question becomes a fixed point, and evaluating its stability can be done with standard linearization techniques.

In this section we show how the Poincar\'e map of BRD in the neighborhood of a given periodic walk in the preference graph can be computed explicitly, and what's more it has a very well-behaved structure---it is a linear map! We can therefore compute its eigenvectors and hence determine the stability of the associated walk in the graph. Specifically, given a walk $c$, we shall construct a matrix $M_c$ which defines the Poincar\'e map along $c$. 

The idea of this matrix is an old one, with precursor ideas in \cite{shapley_topics_1964,krishna_convergence_1998,hofbauer1998evolutionary,gaunersdorfer1992time,gaunersdorfer_fictitious_1995}. In their notable work, \cite{krishna_convergence_1998} derive a form of $M_c$ that is essentially equivalent to ours, albeit with some key differences. The dimensions of their matrix scales with the (potentially unbounded) length of the walk, while ours is bounded by the number of strategies. We also express the structure of $M_c$ as a product of much smaller matrices, defined by the component preference graph arcs in the walk, which allows us to immediately derive some key properties. These component matrices are called the \emph{arc matrices}.


\begin{defn}[Arc matrix] \label{def: arc matrix}
    Let $a = \arc{p}{q}$ be an arc in the preference graph, where $p$ and $q$ differ in the strategy of player $i$, and $s_1$ and $s_2$ are the strategies that player $i$ switches between. Then the \emph{arc matrix} of $a$, written $M_a$, is the $\sum_i^N|S_i|\times \sum_i^N|S_i|$ matrix
    \[
    M_a = I + u(p)c^T_a
    \]
    where $c_a = (e^i_{s_1} - e^i_{s_2})/(u^i_{s_2}(p) - u^i_{s_1}(p))$.
\end{defn}
Note that the denominator $u^i_{s_2}(p) - u^i_{s_1}(p)$ is exactly the weight $W_{p,q}$ on the arc $a$ in the preference graph, so equivalently $c_a = (e^i_{s_1} - e^i_{s_2})/W_{p,q}$.

The meaning of the arc matrix is straightforward to understand. Suppose we are given a point $w$ in payoff space, where we assume that $\arg\max w = p$, and the first BRD switch on the BRD path from $w$ switches to $q$. Then the arc matrix computes the value of $w$ at the time $\tau$ when this BRD switch occurs.
The following lemma formalizes this.
\begin{lem} \label{arc matrix property}
    Suppose that $w$ is a point in payoff space with $\arg\max w = p$, where $p$ is a pure profile. Suppose that after some time $\tau$, BRD switches from $p$ to $q$. Then the value $w(T)$ of $w$ at time $T$ is exactly
    \[
    w(T) = M_a w
    \]
\end{lem}
\begin{proof}
    Let $i$ be the switching player at time $\tau$, and let $s_1$ and $s_2$ be the strategies they switch between. Observe that the switching time $T$ is defined by the point where when $w^i_{s_1} + u^i_{s_1}(p) \tau = w^i_{s_2} + u^i_{s_2}(p) T=\tau$, that is $\tau = (w^1_{s_1} - w^1_{s_2})/(u^1_{s_2}(p) - u^1_{s_1}(p))$ which is exactly $c_a^T w$. Before the switch occurs, each player $i$ accrues payoff at the rate given by $u^i(p)$, and this occurs for time $\tau$. Hence $w(\tau) = w + u(p) \tau = w + u(p)c_a^T w = (I + u(p)c_a^T) w = M_a w$.
\end{proof}

This lemma provides the foundation to build the Poincar\'e matrix. Simply put, given a finite periodic sequence of BRD switches, we obtain an expression for the BRD path as a finite product of the arc matrices along the associated walk in the preference graph.

\begin{defn}[Poincar\'e matrix]

    Let $c = p_1,p_2,\dots,p_K$ be a periodic walk of length $K$ in the preference graph. The \emph{Poincar\'e matrix} $M_c$ of $c$ is defined as the composition of arc matrices
    \[
    M_c := M_{a_K} M_{a_{K-1}}\dots M_{a_1}
    \]
    where $a_1,a_2,\dots,a_K$ are the arcs between the profiles in $c$.
\end{defn}

Following the intuition of Lemma~\ref{arc matrix property}, this fairly simple construction is enough to give a representation of the Poincar\'e map of BRD from some hyperplane, with respect to the walk $c$.

\begin{corol} \label{poincare matrix property}
Let $c = p_1,p_2,\dots,p_K$ be a periodic walk, and $M_c$ its Poincar\'e matrix. Let $w$ be a point on the set $H = \arg\max^{-1}p_1 \cap \arg\max^{-1}p_K$ (where $\arg\max^{-1}p := \{w : p\in \arg\max w\}$). Suppose that from $w$ the play sequence of BRD follows the walk $c$ until it again intersects $H$, at a point $w'$. Then
    \[
    w' = M_c w
    \]
\end{corol}
\begin{proof}
    Follows from iteration of Lemma~\ref{arc matrix property}.
\end{proof}

Our interest in the Poincar\'e matrix $M_c$ is due to its ability to determine whether the periodic walk $c$ is stable under BRD. Our main theorem of this section is that the stability can be characterized by analyzing the eigenvalues and eigenvectors of $M_c$.

\begin{thm} \label{thm: stability test FP}
    Let $c$ be a periodic walk in the preference graph, with $M_c$ its Poincar\'e matrix. If $c$ is stable under BRD, then there exists a dominant real eigenvector $\lambda$ of dimension 1, with $\lambda > 1$, with a $\lambda$-eigenvector $\hat w$ where for each $i$ define $\hat w_i := M_{a_i}M_{a_{i-1}}\dots M_{a_1} \hat w$ and two properties hold: (1) $\arg\max \hat w_i = p_i$ and (2) $\hat T_i := c^T_{a_i}\hat w_i \geq 0$. Conversely, if $\lambda$ is a dominant eigenvalue of $M_c$ with eigenvector $\hat w$ where properties (1) holds, and (2) holds strictly\footnote{If (2) held only non-strictly, then the eigenvector lies on the boundary of the cone of points which follow $c$ forever. While nearby points converge to the eigenspace, they need not stay in this cone while doing so, so it is possible that no points strictly follow the path $c$ forever.}, then $c$ is stable under BRD and $\lambda > 1$.
\end{thm}
\begin{proof}
    First, assume that there is a stable BRD walk whose sequence of play is $c$. This implies that there exists an open set of points whose sequence of play is $c$. Any such BRD path intersects the set $H$ which is the intersection of the regions $\arg\max^{-1} c_1$ and $\arg\max ^{-1}c_K$, and so there exists some points (indeed, an open set of points) which converge to this sequence and which lie on $H$. Let $w$ be one such point. Because points sufficiently close to the periodic walk $c$ have a unique BRD trajectory, we can assume that the trajectory from $w$ is unique without loss of generality.
    
    We then define a sequence $w^{(i)}$ where $w^{(i)}$ is the $i$th intersection of the BRD trajectory from $w$ with $H$. By Corollary~\ref{poincare matrix property}, $w^{(i+1)} = M_c w^{(i)}$. In other words, the sequence $w^{(i)}$ is a discrete-time linear dynamical system governed by the matrix $M_c$. The long-run behavior of the BRD trajectory from $w$ is hence determined by the maximal eigenvalues of $M_c$.

    (\emph{Claim:} there is a maximal eigenvector $\lambda$ which is real, and whose eigenvector $\hat w$ lies in $H$)

    Given $M_c$, denote by $K$ the set of points $w$ in $H$ where the BRD sequence of play $w$ follows $c$ forever. It is easy to see that $K$ is closed convex cone, as the constraints defining a BRD path are linear, given by the argmax. Further, because $c$ is persistent, $K$ contains an open set, and so it is also proper. By the definition of $K$, $M_c$ leaves $K$ invariant, that is $M_c K \subseteq K$. By the Generalized Perron-Frobenius Theorem \cite{berman1994nonnegative}, there is a real eigenvalue $\lambda$ that is maximal ($ = \rho(M_c)$) and whose eigenvector lies in $K$.
    
    (\emph{Claim}: $\lambda$ is simple and dominant and $\lambda > 1$).

    The eigenvalue $\lambda$ is simple if the dimension of its eigenspace is one, and it is dominant if all other eigenvalues have a strictly smaller magnitude. If this were false, then the eigenspace spanned by maximal-magnitude eigenvectors would have dimension greater than one. In particular, no eigenvector is isolated, so the neighborhood around $\hat w$ intersects the maximal eigenspace and contains infinitely many vectors in this space. This contradicts the stability of $c$, which requires that the limit points be locally isolated. 

    Suppose that $|\lambda| < 1$. Then the limit of every point in an open set around $w$ is the origin, but this set has measure zero and so contradicts the fact that BRD is volume-conserving in payoff space \cite{swenson_best-response_2018}. Now we want to show that $\lambda > 1$. Suppose for contradiction that $\lambda = 1$. Then all eigenvalues of $M_c$ are bounded by one, and so all orbits are bounded. Under iteration, our initial open set $U$ with positive measure converges to a bounded region lying on the 1-dimensional $\lambda$-eigenspace, which again contradicts volume conservation of BRD in payoff space.
    
    

    (\emph{Claim}: Properties (1) and (2) hold.)
    For any $i$, define $T^i_j$ to be the time spent playing the profile $c_j$ in the $i$th iteration around $c$ from $w$. By Lemma~\ref{arc matrix property}, this is exactly
    \[
    T^i_j = c^T_{a_j}M_{a_j}M_{a_{j-1}}\dots M_{a_1} w^{(i)}
    \]
    Because the sequence of play follows $c$ repeatedly, this value is always positive, and further it is bounded below. Hence in the limit $\hat w$, the time interval
    \[
    \hat T_j = c^T_{a_j}M_{a_j}M_{a_{j-1}}\dots M_{a_1} \hat w
    \]
    must be non-negative, as it is the limit of a positive sequence.


    For the converse, suppose that $\lambda$ is a dominant eigenvalue that is real and at least 1, and $\hat w$ is a $\lambda$-eigenvector satisfying (1) and (2), where (2) holds strictly. Properties (1) and (2) tells us that the BRD path from $\hat w$ follows the walk $c$, and that the time $\hat T_i$ spent at each profile in $c$ on this BRD path is strictly positive, and hence bounded below. Finally, because the times are bounded below, there is an $\epsilon$-neighborhood of $\hat w$ where the trajectory also follows $c$ forever, so it is persistent. In other words, $\hat w \in\intr K$. As $\lambda$ is dominant, then all points in a sufficiently small neighborhood around $\hat w$ where all point converge to the eigenspace of $\hat w$, giving stability.
\end{proof}

Note that the Poincar\'e map of a cycle $c$ depends on where we decide to `start' the cycle---which profile we call the first one. Hence if $c_i = p_i,p_{i+1 \mod K},,\dots,p_{i + K - 1 \mod K}$, the matrices $M_c$ and $M_{c_i}$ will be different. However, these matrices have the same eigenvalues, and their eigenvectors related in a natural way.

\begin{lem}
    Let $c = p_1,p_2,\dots,p_K$ be a periodic walk, and let $c_i = p_i,p_{i+1},\dots,p_K,p_1,\dots,p_{i-1}$ be a rotation of that periodic walk from the $i$th index. Then $M_c$ and $M_{c_i}$ have the same set of eigenvalues.
\end{lem}
\begin{proof}
    Suppose $\lambda$ is an eigenvalue, and $\hat w$ a $\lambda$-eigenvector. Then define $M_{c^2_i} := M_{a_K} M_{a_{K-1}}\dots M_{a_i}$ and $M_{c^1_i} := M_{a_{i-1}} M_{a_{K-1}}\dots M_{a_1}$. Observe that $M_c = M_{c^2_i} M_{c^1_i}$ and $M_{c_i} = M_{c^1_i} M_{c^2_i}$. Let $\hat w_{c_i} := M_{c^1_i} \hat w$. $ \lambda \hat w = M_{c}\hat w = M_{c^2_i} M_{c^1_i} \hat w = M_{c^2_i} \hat w_{c_i}$ and so $\lambda \hat w_{c_i} = \lambda M_{c^1_i} \hat w = \lambda M_{c^2_i} M_{c^1_i} \hat w_{c_i} = \lambda M_{c_i} \hat w_{c_i}$. Hence $\lambda$ is an eigenvalue of $M_{c_i}$, and $\hat w_{c_i}$ is a $\lambda$-eigenvector.
\end{proof}

\subsection{4-cycles} \label{sec: 4cycle}

In this section we discuss the special case of 4-cycles. Here, a simpler lemma suffices to characterize stability, and it gives useful intuition of why Nash equilibrium convergence is common in small games. The lemma involves an explicit construction of the Poincar\'e map, and provides a useful demonstration of this technique on a small game.

\begin{lem}[4-cycle stability]
    A 4-cycle is persistent if and only if it spans the support of a Nash equilibrium.
\end{lem}
\begin{proof}
    First, observe that any cycle requires at least two deviations by any player who participates in it. Hence in a 4-cycle, only two-players participate, and the profiles $c_1,c_2,c_3,c_4$ on the 4-cycle must span a $2\times 2$ subgame of the game. This allows us to treat the game as if it were $2\times 2$, writing $c= (s_1,r_1)$, $c_2=(s_1,r_2)$, $c_3=(s_2,r_2)$ and $c_4=(s_2,r_1)$. Let $\alpha_1,\alpha_2,\alpha_3,\alpha_4$ be the weights on the arcs between the nodes. By strategic equivalence, we can assume the payoffs are normalized so that the payoffs for strategies $s_1$ and $r_1$ are always zero (and any other players who do not deviate on this cycle receive zero payoff).

    This $2\times 2$ subgame necessarily contains a fixed point $\hat x$, whose distribution for players 1 and 2 is respectively $((\frac{\alpha_4}{\alpha_2 + \alpha_4},\frac{\alpha_2}{\alpha_2 + \alpha_4}),(\frac{\alpha_3}{\alpha_1 + \alpha_3},\frac{\alpha_1}{\alpha_1 + \alpha_3}))$. By normalization, the expected payoff at this point to all players is zero.
    
    Now let $w_0$ be a point on the intersection $H$ of $\arg\max^{-1}(c_1)$ and $\arg\max^{-1} c_4$. By normalization, $w_0 = ((0,0),(0,-x))$ for some $x$. The next point of indifference, between $c_1$ and $c_2$, occurs at point $w_1$ after time $T_1$. These can be computed explicitly, as $T_1 = x/\alpha_1$, and $w_1 = ((0,-\alpha_4 T_1),(0,0))$. Repeating this same calculation gives us: $T_2 = \alpha_4 T_1/\alpha_2$, and $w_2 = ((0,0),(0,\alpha_1 T_2))$; $T_3 = \alpha_1 T_2/\alpha_3$, and $w_3 = ((0,\alpha_2 T_3),(0,0))$ and finally $T_4 = \alpha_2 T_3/\alpha_4$, and $w_4 = ((0,0),(0,-\alpha_3 T_4))$.

    By definition of the Poincar\'e matrix, $w_4 = M_c w_0$. In fact, $w_4 = w_0$, because
    \[
    \alpha_3 T_4 = \frac{\alpha_3 \alpha_2 T_3}{\alpha_4} = \frac{\alpha_2 \alpha_1 T_2}{\alpha_4} = \alpha_1 T_1 = x\ .
    \]
    We conclude that $M_c$ acts as the identity on the subspace corresponding to this $2\times 2$ subgame.

    Suppose that $g$ is an alternative strategy, assumed for player 1, where $g$ receives payoff $y$ in profiles $c_1$ and $c_2$ and payoff $z$ in the profiles $c_3$ and $c_4$. If the \emph{net} payoff for $g$ over the cycle is positive, then eventually the BRD path must switch away from the cycle. More generally, the persistence of this cycle is exactly determined by the presence of a strategy for some player which earns positive net payoff relative to the payoffs on the cycle. The net payoff for $g$ is exactly
    \[
    (T_1 + T_2) y + (T_3 + T_4) z = (\frac{x}{\alpha_1} + \frac{\alpha_4 x}{\alpha_1\alpha_2})y + (\frac{\alpha_4 x}{\alpha_2\alpha_3} + \frac{x}{\alpha_3})z = x\frac{\alpha_2 + \alpha_4}{\alpha_1\alpha_3} (\alpha_3 y + \alpha_1 z)
    \]
    The sign of the net payoff to $g$ is determined by $\alpha_3 y + \alpha_1 z$, but expression also determines the sign of the payoff for $g$ at the fixed point $\hat x$, which is
    $\frac{\alpha_3 y + \alpha_1 z}{\alpha_1 + \alpha_3}$.
    A symmetric construction works for player 2. Finally, we must do the same thing for some other player $i$, who does not deviate in this cycle. For any arbitrary strategy $g$, they receive payoffs $a$, $b$, $c$ and $d$ at $c_1$, $c_2$, $c_3$ and $c_4$ respectively, which gives them a payoff at the fixed point $\hat x$ given by:
    \[
    \frac{1}{(\alpha_1 + \alpha_3)(\alpha_2 + \alpha_4)}(a \alpha_2\alpha_3 + b \alpha_3\alpha_4 + c \alpha_1\alpha_4 + d \alpha_1\alpha_2)
    \]
    Similarly, the payoff along the 4-cycle for strategy $g$ is given by
    \[
    a \frac{x}{\alpha_1} + b\frac{\alpha_4 x}{\alpha_1\alpha_2} + c \frac{\alpha_4 x}{\alpha_2\alpha_3} + d \frac{x}{\alpha_3} = x\alpha_1\alpha_2\alpha_3 (a \alpha_2\alpha_3 + b \alpha_3\alpha_4 + c\alpha_1\alpha_4 + d\alpha_1\alpha_2)
    \]
    As before, the bracketed term necessarily has the same sign as the fixed point payoff for $g$. We conclude that no player ever deviates from this cycle if and only if no other strategy provides strictly better than equilibrium payoff, which is the condition for this fixed point to be a Nash equilibrium.
\end{proof}

Note that the 4-cycle is only persistent, not stable (recall Definition~\ref{def: periodic play}). This is because, in payoff space, each trajectory following the 4-cycle is periodic. In other words, all maximal eigenvalues of the Poincar\'e matrix have magnitude 1. This can be equivalently thought of as a zero-sum property---every $2\times 2$ game whose preference graph is a cycle is strategically equivalent to a zero-sum game \cite{biggar_graph_2023}. The strategy space of these games is the unit square, and under the replicator dynamic every interior point (other than the Nash equilibrium itself) follows a periodic orbit circling the Nash equilibrium. Like BRD, the time-average of the replicator dynamic converges.

Berger \cite{berger_fictitious_2005} proved that FP always converges to Nash equilibrium in $2\times n$ games. His analysis showed precisely that almost every starting point converges to either a pure Nash equilibrium or a 4-cycle, and hence to a mixed Nash equilibrium. Krishna and Sj\"ostrom \cite{krishna_convergence_1998} show that 4-cycles are special in this sense: in a two-player game, any cycle longer than four does not converge to Nash equilibrium, explaining the distinction between $2\times n$ games and $3\times 3$ games such as Shapley's.

\section{Stable cycles under the replicator dynamic} \label{sec: RD}

So far, our analysis of stable cycles has been restricted to BRD. There, the piecewise linearity of trajectories allowed us to analyze stability using linear algebra. The highly nonlinear structure of the replicator dynamic makes it difficult to see how to extend this reasoning. More generally, results identifying cyclic patterns of behavior for the replicator are far less well-developed than for BRD. Here we prove an important theorem: we can reduce the problem of studying stable cycles of the replicator to studying stable cycles of BRD. More specifically, if we have a stable periodic walk under BRD, then it is stable under the replicator dynamic. 

\begin{thm} \label{thm: RD stability = FP stability}
    Let $c$ be a periodic walk in the preference graph. Then $c$ is stable  under BRD if and only if it is stable under the replicator dynamic.
\end{thm}
\begin{proof}
    To begin, define $H$ to be the cone in payoff space defined by $\arg\max^{-1} c_1 \cap \arg\max^{-1}(c_K)$. Suppose that $w$ is a point in $H$ such that the sequence of play of either RD or BRD follows the cycle $c$ and returns to again intersect $H$. Then we denote by $RD_H(w)$ and $BRD_H(w)$ respectively the subsequent intersection of the trajectory of $w$ with $H$. If BRD follows the $c$ from $w$, then by Theorem~\ref{poincare matrix property} we know that $BRD_H(w) = M_c w$. The critical fact---which will require the most technical arguments of the proof---is the following: as $\|w\|$ grows large, the trajectories of $RD_H$ and $BRD_H$ approach each other.

    \begin{prop}[Asymptotic Linearity] \label{replicator - FP similarity}
    Let $w_0$ be a point where all points in the $\delta$-neighborhood of $w_0$ follow $c$ forever, and are bounded away from the boundary, for some sufficiently small $\delta > 0$. Then $\| RD_H(w_0) - BRD_H(w_0)\| \to 0$ as $\|w_0\| \to \infty$.
    \end{prop}
\begin{proof}[Proof of Proposition~\ref{replicator - FP similarity}]
    First, we will express our initial point as $w_0 = \gamma w$, where $\|w\| = 1$. Hence we can express our goal as showing that $\| RD_H(\gamma w) - BRD_H(\gamma w)\|/\gamma \to 0$ as $\gamma \to\infty$. We then define a sequence of time points $T_1,T_2,\dots,T_K$ where BRD switches occur on the path from $w$ to $BRD_H(w)$.  By linearity, starting from $\gamma w$ the BRD switches occur at the same times but scaled by $\gamma$: $\gamma T_1, \gamma T_2,\dots,\gamma T_m$. For each $v$ in $\| v - w\| < \delta$, the size of each $T_i$ is bounded below, and we let $T_{min}$ be a lower bound for this set. Next, we divide the total time $\gamma T_{total} = \gamma \sum_i T_i$ into a finite number $\rho$ of intervals of length $\tau \gamma$, where $\tau$ is some constant with $\tau < T_{min}$. This implies that there is always \emph{at most one} player who switches strategy in that interval under BRD. Note that (crucially) $\rho$ is independent of $\gamma$. For any starting point $\gamma w'$ in the $\delta$-neighborhood of $\gamma w$, we get the same collection of $\rho$ intervals, with the length of each interval scaling with $\gamma$. We assume that we pick $\tau$ such that no BRD switch from $w'$ lands exactly on the endpoint of an interval, and we define $\Delta > 0$ to be a lower bound for the distance from a BRD switch and the endpoint of one of these intervals.
    

    Let $f_0,f_1,\dots, f_\rho$ be the sequence of points reached by the BRD path from $w$ after elapsed times $0,\tau,2\tau,\dots,\rho\tau$, with $f_0 = w$ and $f_\rho = BRD_H(w_0)$. 
    Similarly, we define $r^\gamma_0,r^\gamma_1,\dots,r^\gamma_\rho$ to be the sequence of points on the replicator trajectory from $\gamma w$ at these times. Our goal is to show that the distance between $\gamma f_\rho$ and $r\gamma_\rho$ vanishes as $\gamma \to\infty$. 
    
    We can break up this distance in terms of the error at each step, so
    \begin{align*}
    \|RD_H(\gamma w) - BRD_H(\gamma w) \| = \|r^\gamma_\rho - \gamma f_\rho\| &=  \left\| \int_{0}^{\rho\tau\gamma} (\rd(\gamma w) - \brd(\gamma w))\dx t\right\| \\
    &\leq \sum_{i=1}^\rho \left\| \int_{0}^{\tau\gamma} (\rd(r^\gamma_i) - \brd(\gamma f_i))\dx t\right\|
    \end{align*}
    where we define $\rd$ and $\brd$ to the defining differential equations of RD and BRD respectively in payoff space, as in Definitions~\ref{def: rd payoff} and \ref{def: brd payoff}. To show this entire sum vanishes as $\gamma$ grows, we will use the fact that, if 'total error' $E = \|r^\gamma_i - \gamma f_i\|$ up to the $i$th step is small, then the additional error from this step
    $\| \int_{0}^{\tau\gamma} \rd(w) - \brd(f) \dx t\|\to 0 $ 
    vanishes as $\gamma$ grows. From there, the fact that the whole sum vanishes will result from induction with $\|r^\gamma_0 - \gamma f_0\| = 0$, and the fact that the number of intervals is a constant in $\gamma$. For simplicity of notation we write $w := r^\gamma_i$, $f := \gamma f_i$ and $E  = \|w-f\|$.
    There are two cases to consider: the first where the $i$th interval contains no BRD switches, and the second where it contains exactly one.

    (\emph{Case 1}: the $i$th interval contains no BRD switches)

    Let $c$ be the strategy profile played in this time interval from $f$. Fix a player $j$ and strategy $s$ for $j$. The gradient of BRD from $f$ is exactly the vector $u^j(c)$ in this interval. On the other hand, $\rd(w)_s = \sum_{p} u^j(p) z_p = \sum_{p_{-j}} u^j(p_{-j}) z_{p;-j}$ as the utility for player $j$ does not depend on $x^j$. Then
\begin{align*}
    \left| \int_{0}^{\tau\gamma} (\rd^j(w)_s - \brd^j(f)_s )\dx t\right| &= \left | \int_{0}^{\tau\gamma} (\rd^j(w)_s - u^j_s(c))\dx t\right | = \left | \int_{0}^{\tau\gamma} (\sum_p u^j_s(p) z_p - u^j_s(c) )\dx t\right | \\
    &= \left | \int_{0}^{\tau\gamma} (u^j_s(c)(1 - \sum_{p\neq c} z_p) + \sum_{p\neq c} u^j_s(p) - u^j_s(c))\dx t\right | \\
    &\leq \sum_{p\neq c} |u^j_s(p) - u^j_s(c)| \int_{0}^{\tau\gamma} z_p\dx t
\end{align*}

Because $c$ is the current pure profile of best responses for $f$, for each profile $p\neq c$, there is at least one player $k$ for whom $p_k$ is not a best response in $f$. Further, because no player switches under BRD within time $\Delta \gamma$ of the current interval from $f$, we get that $f^k_{c_k} - f^k_{p_k} > \mu \Delta \gamma$, as $\mu$ is the least difference in payoff between any pair of strategies. Because $\|w - f\| < E$, we get
\[
z_p = \prod_i x^i_{p_i} \leq x^k_{p_k} = \frac{\exp(w^k_{p_k})}{\sum_s \exp{(w^k_{s})}} = \frac{\exp(w^k_{p_k} - w^k_{c_k})}{1 + \sum_{s\neq c_k} \exp{(w^k_{s} - w^k_{c_k})}} \leq \exp(w^k_{p_k} - w^k_{c_k}) \leq \exp(-\mu\Delta \gamma + 2E)
\]

Finally, we can bound the integral $\int_0^{\tau\gamma} z_p\dx t$ by $\tau\gamma \exp(-\mu\Delta \gamma + 2E)$ and conclude, for any $j$ and $s$,
\begin{equation}
\label{bound}
\left| \int_{0}^{\tau\gamma} (\rd^j(w)_s - \brd^j(f)_s )\dx t\right| < O(\gamma \exp(-\alpha \gamma + 2E))
\end{equation}
for some $\alpha > 0$. For $E$ sufficiently small, the quantity inside the exponent is negative, and so it vanishes as $\gamma \to\infty$..

(\emph{Case:} A single BRD switch in the $i$th interval.)

We assume that the switch occurs at time $i\tau\gamma + q\gamma$, with $0 < q < \tau$. We let $c_1$ and $c_2$ be the profiles we switch from and to respectively, and assume w.l.o.g. that the switching player is player 1. We let $a$ and $b$ be the strategies played by player 1 in profiles $c_1$ and $c_2$ respectively. Note, importantly, that for player 1 the calculation is the same as before. This is because the payoff to player 1 depends only on the other players, none of whom switch in this interval. Hence we need only bound the deviation between RD and BRD over this interval for some player $j\neq 1$ and strategy $s$:
\begin{align*}
    \left| \int_0^{\tau\gamma} (\rd^j(w)_s - \brd^j(f)_s )\dx t\right|
    &= \left| \int_{0}^{\tau\gamma} \sum_p u^j(p)z_p\dx t - u^j(c_1)q\gamma - u^j(c_2) (\tau\gamma - q\gamma) \right| \\
    & \leq \sum_{p\neq c_1,c_2} u^j(p)\left|\int_{0}^{\tau\gamma} z_p\dx t\right | + u^j(c_1)\left | \int_{0}^{\tau\gamma}z_{c_1} \dx t - \gamma q\right |  \\ &+ u^j(c_2)\left|\int_{0}^{\tau\gamma}z_{c_2} \dx t- (\tau\gamma - \gamma q)) \right|
\end{align*}

We shall show that, if $E$ is small, then all of these terms vanish as $\gamma \to \infty$. We begin with the first term. This term can be bounded in much the same way as the previous case (without switches). In this interval, the best response of any player $j\neq 1$ is $c^j_1$ (which is the same as $c^j_2$). For player 1, the best-response could be either $a$ or $b$. For any profile $p$ that is not $c_1$ or $c_2$, we either have some player $j$ playing a strategy not equal to $c^j_1$, or player 1 playing a strategy other than $a$ or $b$. In each case, there is a player $k$ playing a strategy $p_k$ where $f^k_{c_k} - f^k_{p_k} > \mu \Delta \gamma$, because we do not switch to that strategy within time $\Delta \gamma$ (either side) of the current interval in the BRD path from $f$. As before, we can bound the expression
\[
\sum_{p\neq c_1,c_2} u^j(p)\left|\int_{0}^{\tau\gamma} z_p\dx t\right | \leq O(\gamma \exp(-\gamma \alpha + 2E))
\] for some constant $\alpha > 0$, and this goes to zero as $\gamma \to\infty$ for small enough $E$. This leaves us with two terms to bound, both of which have essentially the same structure and for which we will use the same argument. We will create an upper bounds for this term using the logistic function, and show that it vanishes as $\gamma$ grows and $E$ vanishes. Observe that

\[
z_{c_1} \leq x^1_{a} = \frac{\exp(w^1_a)}{\sum_s\exp(w^1_s)} \leq \frac{\exp(w^1_a)}{\exp(w^1_a) + \exp(w^1_b)} = \frac{\exp(w^1_a - w^1_b)}{1 + \exp(w^1_a - w^1_b)}
\]

So far, it is not clear how to proceed, because we don't know how to integrate this function, as $w^1_a(t)$ evolves according to the replicator equation. However, reusing the idea of this theorem, we can \emph{approximate} it by the flow of BRD, which we \emph{can} integrate here. First, recall that for any $t\in [0,\tau\gamma)$, $w^1_a(t) = w^1_a(0) + \int_0^t\rd^1_a(w)\dx t$. Because this is player 1, we know two additional facts: (a) $u^1(c_1)$ and $u^1(c_2)$ are equal, because they differ only in the strategy of player 1, and (b) the bound of \eqref{bound} holds in this interval. Using (a), we can write:
\[
w^1_a(t) = w^1_a(0) +  \int_{0}^{t} \rd^1_s(w)\dx t = w^1_a(0) + \int_0^t\rd^j_a(w) - \brd^j_a(f)\dx t + t u^1(c_1)
\]
where we used the fact that the gradient of BRD is a constant value $u^1(c_1)$. By the bound of \eqref{bound}, $\left| \int_{0}^{\tau\gamma} (\rd^j(w)_s - \brd^j(f)_s )\dx t\right| < K\gamma \exp(-\alpha \gamma + 2E)$ for some constant $K$. We call this quantity $L(E,\gamma)$, and recall that $L(E,\gamma)\to 0$ as $\gamma\to\infty$. Then we get a bound that holds for all $t$ in this interval:
\[
|w^1_a(t) - w^1_b(t) - (f^1_a(0) - f^1_b(0)- t(u^1(c_1) - u^1(c_2)))| \leq 2|E + L(E,\gamma)|
\]

This inequality gives us bound for our integral which we can anti-differentiate:

\begin{align*}
    \left | \int_{0}^{\tau\gamma} \frac{\exp(w^1_a - w^1_b)}{1 + \exp(w^1_a - w^1_b)} \dx t \right|
    &\leq \left | \int_{0}^{\tau\gamma} \frac{\exp(f^1_a(0) - f^1_b(0)+ t(u^1(c_1) - u^1(c_2)) + 2E + 2L(E,\gamma))}{1 + \exp(f^1_a(0) - f^1_b(0)+ t(u^1(c_1) - u^1(c_2)) - 2E - 2L(E,\gamma))} \dx t \right| \\
    &\leq \exp(4E + 4L(E,\gamma)) \left | \int_{0}^{\tau\gamma} \frac{\exp(X + Yt)}{1 + \exp(X + Yt)} \dx t \right|
\end{align*}
where $X + Yt = f^1_a(0) - f^1_b(0)+ t(u^1(c_1) - u^1(c_2)) - 2E - 2L(E,\gamma)$. 

Here we observe that, when the terminals are symmetric around the $x$-intercept, this integral can be integrated in a straightforward way:
\[
\int_{-b}^b \frac{\exp(at)}{1+ \exp(at)}\dx t = \left [(1/a)\log(\frac{1 + \exp(at)}{1 + \exp(-at)}) \right]_{-b}^b= b
\]
We will evaluate our integral by shifting the terminals to center around the $x$-intercept. If not for the error term $2E + 2L(E,\gamma)$, this would be $q\gamma$, because $f^1_a(0) - f^1_b(0) + q(u^1(c_1) - u^1(c_2)) = 0$. Recall that this is the definition of $q$---it is the point where the BRD switch occurs for player 1 on the path from $f$. With the error terms, the midpoint is at $k := q\gamma + (2E + 2L(E,\gamma))/Y$. We can therefore bound the integral by a pair of integrals, where the first has terminals $[0,2k]$ and the latter has terminals $[2k,\tau\gamma]$. The former is symmetric around its midpoint $k$, and evaluates to $k$. We arrive at a new bound of
\[
\left | \int_{0}^{\tau\gamma}z_{c_1} \dx t - q\gamma \right | \leq \left |\exp(4E + 4L(E,\gamma)) \left (k + \left| \int_{2k}^{\tau\gamma} \frac{\exp(X + Yt)}{1 + \exp(X + Yt)} \dx t \right|\right) - q\gamma \right |
\]
Because $X + Yt = 0$ when $t=k$, and the integrand is always decreasing in $t$, the integrand is bounded above on the region $[2k,\tau\gamma]$ by $\exp(-k)/(1+\exp(-k)) \leq \exp(-k) \leq \exp(-q\gamma)$. Hence the integral is bounded by $(\tau\gamma - 2k)\exp(-q\gamma)$, which vanishes as $\gamma \to\infty$. It remains only to show that $\exp(4E + 4L(E,\gamma))k  - q\gamma= \exp(4E + 4L(E,\gamma))(q\gamma + (2E + 2L(E,\gamma))/Y) - q\gamma$ vanishes as $\gamma \to \infty$. Near zero, $\exp$ vanishes at the same speed as its argument, which vanishes exponentially quickly compared to any polynomial in $\gamma$. Hence this term is bounded by a quantity whose limit goes to zero, and so it also converges to zero. The same argument, \emph{mutatis mutandis}, works for the $z_{c_2}$ term as well. This completes the proof.
\end{proof}

    With Proposition~\ref{replicator - FP similarity} in hand, we can tackle the remainder of the proof. As before, we define $K$ to be the closed convex cone in $H$ where the BRD sequence of play $w$ follows $c$ forever.

    ($\Rightarrow$) Suppose that $c$ is a stable cycle for BRD. Then it is stable for RD.

    By Theorem~\ref{thm: stability test FP}, if $c$ is stable for BRD then $M_c$ has a unique maximal eigenvalue $\lambda > 1$  with eigenvector $\hat w$ in $K$.
    Let $w_k := RD_H^k(w)$. Our first step is to show that $\|w_k\|\to\infty$ as $k\to\infty$, if $\|w\|$ is initially sufficiently large. Write $RD_H(w) = M_c w + E(w)$ for some error term $E$, by Proposition~\ref{replicator - FP similarity}, which vanishes as $\|w\| $ grows. Let $\hat z$ be the \emph{dual eigenvector} of $\hat w$, that is, $\hat z^T M_c = \lambda \hat z^T$. $\hat z^T$ lies in the interior of the dual cone of $K$, so for each $w\in K$ we have $\hat z^Tw > 0$. Finally, note that because $\hat z^T w/\|w\|$ is bounded, the iterates $\hat z^T w_k$ grows to infinity if and only if $\|w_k\|\to\infty$. Then
    \[
    \hat z^T RD_H(w) = \hat z^T( M_c w + E(w)) = \lambda \hat z^Tw + \hat z^T E(w)
    \]
    Because $E(w)$ vanishes for large $w$, if $\|w\|$ is sufficiently large we get $\hat z^T E(w) < \nu z^T w$ and so $
    \hat z^T(RD_H(w)) > (\lambda - \nu) z^T w $ for $\nu < \lambda - 1$, show $\|w\|$ grows to infinity. Hence after this point, as $\|w\|$ grows, by Proposition~\ref{replicator - FP similarity} the orbits of BRD and RD get arbitrarily close together
    The remainder of the proof follows from straightforward arguments by the fact that the iterates of BRD converge to $\hat w$ in $K$. This tells us that $c$ is stable under RD, with the same limiting eigenvector $\hat w$.

    ($\Leftarrow$: If $c$ is stable under RD, it is stable under BRD.)

    Suppose that $c$ is stable under RD, so there is an open set $U$ of starting points such that the sequence of play under RD follows $c$, and these share the same limit in the normalized form $w/\|w\|$. Suppose first that the orbits of all points in $U$ are bounded. Then all points in $U$ have limit points in $H$. By stability, the limit points are locally isolated, and so they cannot have be a full-measure set. This contradicts the fact that the replicator preserves volume in payoff space, and so cannot converge to a lower-dimensional set like a periodic orbit. 
    We conclude that the sequence $RD_H^k(w)$ of iterates is unbounded.

    Let $u_k = RD_H^k(w) / \|RD_H^k(w)\|$ be the normalized form of the sequence, which lies on the unit sphere $S^{n-1} \cap H$. Because $c$ is persistent for RD, this sequence lies entirely in this set. Let $z$ be a limit point of this sequence, which exists by compactness. Now let $\Phi$ be the projected form of the BRD map, so $\Phi(u) = \frac{M_c u}{\|M_cu\|}$. As $\|x\|\to\infty$, $\Phi(w_k)\approx RD(w_k)/\|RD(w_k)\|$ by Proposition~\ref{replicator - FP similarity}, and so the limit points of both maps are the same. However, any limit points for $\Phi$ must lie in the closed cone $K\subseteq H$, because their BRD iterates remain in $H$ forever.
    As the open set $U$ approaches this limit under RD, there must be an open set which lies within $K$, so this cycle is persistent for BRD. Because the limit points are the same, and the limit points are isolated for RD, they are isolated for BRD as well, giving us that $c$ is stable for BRD.
\end{proof}







\section{Cyclic sink equilibria are attractors} \label{sec: cyclic sink equilibria}

In the previous sections, we developed an eigenvector-based computational test for when a walk in the preference graph defines a persistent walk under BRD, and showed how stable periodic plays of BRD define stable periodic plays for the replicator dynamic. In this section we give a structural (that is, graph-theoretic) criterion for a walk $c$ which, when satisfied, implies that $c$ is stable under BRD and the replicator. This condition is inspired by the key property of Shapley's \cite{shapley_topics_1964} and Jordan's games \cite{jordan_three_1993}, when viewed as graphs (Figure~\ref{fig:shapley and jordan}). There, the walks in question have no outgoing arcs (that is, are in both cases a sink connected component of the preference graph). We show this is generally true: whenever a walk has no outgoing arcs, then it is stable under BRD, and hence stable under RD. More strongly, we will show it that the cycle defines an attractor of the replicator dynamic (it is asymptotically stable). 

Before we prove this theorem, several remarks are warranted. First, if a walk has the property that the outgoing arc from a node is always unique, then it must be exactly a cycle (it can have no repeated nodes). From now on, we will describe this as \emph{stable cycle}, rather than simply a stable walk. Additionally, this cycle must consist of best-response (rather than only better-response) deviations, so this cycle is simultaneously a sink equilibrium of the \emph{best-response graph}.

For BRD, the fact that a cycle with no outgoing arcs is persistent is intuitive---BRD switches occur along arcs in the preference graph, and so the cycle is followed repeatedly from an open set. The challenge of this proof involves showing stability, which implies uniqueness and dominance of the eigenvector. We show this through an application of the Generalized Perron-Frobenius theorem, by showing that the orbits of any point in the interior are unbounded.



\begin{thm} \label{thm: sink equilibria are stable}
    Suppose a cycle $c$ is a sink equilibrium (it has no outgoing arcs) in an $N$-player game. Then the union of the strategy spaces of the 1-dimensional subgames defined by the arcs collectively form an attractor of the replicator dynamic.
\end{thm}
\begin{proof}
    We begin by handling a special case, which occurs when the cycle has length exactly four. Recalling Section~\ref{sec: 4cycle}, any such cycle can only involve two players, and must span a $2\times 2$ subgame of the game. Because $c$ is a sink equilibrium, we know that this subgame is a sink equilibrium, and this implies it is an attractor \citep{biggar_replicator_2023}\footnote{Note that in this special case, $c$ is persistent under RD but not stable, because there is are always infinitely many concentric periodic orbits in these games \cite{mertikopoulos_cycles_2018}.}.
    
    Now, suppose the cycle has length $K > 4$ (in a preference graph, every cycle has length at least four). We will show that $c$ is stable under BRD, and hence stable under RD by Theorem~\ref{thm: RD stability = FP stability}. Our method will involve the strong form of the Generalized Perron-Frobenius theorem. As in the previous two theorems, we let $H$ be the set consisting of the intersection $\arg\max^{-1} c_1 \cap \arg\max ^{-1}(c_K)$. The set $H$ is a proper convex cone, that is non-empty. Further, persistence of $c$ is immediate---because every BRD path follows a path in the preference graph, any point which begins on $c$ remains on $c$. This means that $H$ is $M_c$-invariant, and so there is an eigenvalue $\lambda = \rho(M_c)$, whose eigenvector $\hat w$ lies in $H$. To prove stability, we need to establish that $\lambda$ is dominant, and $\hat w$ lies in $\intr H$ (Theorem~\ref{thm: stability test FP}). Using the Perron-Frobenius theorem, this will follow from two facts, which we prove in lemmas below: (1) we show the orbits of points in $\intr H$ grow unboundedly, and (2) from any point $w$ on the boundary of $H$, the orbit of $H$ either enters $\intr H$ or is bounded. This first claim establishes that $\rho(M_c) > 1$, while the second establishes that on the boundary $\rho(M_c|_{\partial H}) \leq 1$. Together, these imply that $\hat w$ lies in $\intr H$, and $\lambda$ is dominant.

    Let $c_1,c_2,\dots,c_K$ (recall $K>4$) be the profiles on the cycle, and let $\alpha_1,\dots\alpha_K$ be the weights on the arcs of the cycle, numbered such that $\alpha_i$ is the arc leaving the $i$th profile on $c$. Because $c$ is a sink equilibrium, every arc into $c$ has nonzero weight, so we let $\mu >0$ denote the minimum weight on any arc into $c$. Let $w$ be any starting point on $H$. 
    Define $T_i$ as the time spent playing $c_i$ on the path from $w$, so $T_i = c_{a_i}^T M_{a_i} M_{a_{i-1}} \dots M_1 w \geq 0$.

    \begin{lem}
        If $w \in H$, then either $M_c w \in \intr H$ or $M_c w = w$.
    \end{lem}
    \begin{proof}
        At each profile $c_i$, any strategy $s$ that is not a best response accrues payoff at a rate that is at least $\mu$ worse than the best response. If $s$ is the strategy we switch to on the $i$th interval, we have $\alpha_i T_i > \mu T_{i-1}$. In particular, if $T_{i-1} > 0$, then $T_{i} > 0$. Next, note that if $T_K > 0$, then $M_c w \in \intr H$, because every strategy other than those in $c_1$ and $c_K$ must have strictly worse payoff, by at least $\mu T_K$. We conclude that if any time $T_i > 0$, then $M_c w\in \intr H$. However, if every $T_i$ is zero, then $M_c$ acts as the identity on $w$---recall that by Definition~\ref{def: arc matrix}, $w^{(i)} = M_{a_i} w^{(i-1)} = w^{(i-1)} + u(c_i) T_i$.
    \end{proof}

    \begin{lem}
        If $w\in \intr H$, then the orbit of $w$ goes to infinity.
    \end{lem}
    \begin{proof}
        Let $w^k$ denote the $k$-th iterate $M_c^k w$, and let $T_i^{k}$ denote the spent playing the $i$th profile $c_i$ from $w^k$, so $T_i^k = c_{a_i}^T M_{a_i} M_{a_{i-1}} \dots M_1 w^k$. Since $w\in \intr H$, each $T_i^k$ is positive. To prove the result, it is sufficient to show that each $T_i^k$ is bounded below, and $T_i^{k}\to\infty$ as $k\to\infty$.

        Given a profile $c_i$, let $s$ be the unique strategy that is a best-response to $c_i$ that is not played in $c_i$, so $c_{i+1} = (s;c_i^{-1})$ where we have assumed without loss of generality that $s$ is a strategy of player 1. Because $c$ is a sink equilibrium, no player switches twice in a row, so player 1 does not switch strategy between times $T_{i-1}$ and $T_i$. Define $\arc{c_{j-1}}{c_j}$ as the previous pair of profiles where the switching player was player 1 (hence $j \neq i-1$). If $w_{(i)}$ is the payoff vector at the beginning of the $i$th interval, we define $\Delta_{i,s}$ to be difference $\max w^1_{(i)} - w^i_s$ between the maximum payoff at that time and the payoff for strategy $s$. Because $s$ must be maximal at the end of the $i$th interval (at this point we switch to $c_{i+1}$), we have that $\alpha_i T_i = \Delta_{i,s}$. While playing $c_{j}$, the player 1 is not playing $s$, so $\Delta_{j+1,s} \geq \alpha_{j-1} T_{j}$. We conclude that for each $i$, there exists a $j \neq i-1$ such that $\alpha_i T_i \geq \alpha_{j}T_{j+1}$.
        
        Note that because $\Delta_{k,s}$ is increasing at any profile $c_k$ where player 1 does not play $s$, equality is only achieved here if both $T_j = T_{i-1}$ and $s$ is played in $c_{j-1}$. Otherwise, we have the stricter inequality $\alpha_i T_i \geq \alpha_{j}T_{j+1} + \mu T_k$ for some $k$. If the cycle is longer than four, the equality case can never happen twice in a row, so if $\alpha_i T_i = \alpha_{j-1} T_{j}$, then $\alpha_{j} T_{j} > \alpha_{k-1} T_k$ for some $k$. 
    
    Applying this inequality repeatedly gives us a sequence of the form
    \begin{equation*} \label{big inequality}
        T_i^{k+1} \geq \frac{\alpha_{j_1 - 1}}{\alpha_i} T^k_{j_1} \geq \frac{\alpha_{j_1 - 1}}{\alpha_i}\frac{\alpha_{j_2 - 1}}{\alpha_{j_1}} T^k_{j_2} \geq \dots \geq \frac{\alpha_{j_1 - 1}}{\alpha_i}\frac{\alpha_{j_2 - 1}}{\alpha_{j_1}} \dots T^k_i + \delta T_{min}
    \end{equation*}
    where $\delta$ is a constant and $T_{min} = \min_i T^k_i$. To show that $T_i$ grows, we want to show its coefficient is at least one. Consider the permutation defined by mapping $i \to j$ by the criteria of $\alpha_i T_i \geq \alpha_{j}T_{j+1}$. This permutation splits into a collection of cycles, each giving us an inequalities in the form of \eqref{big inequality}, as in
    \[
    T_{i_1}^{k+1} > K_1 T^k_{i_1},\quad T_{i_2}^{k+1} > K_2 T^k_{i_2}, \quad \dots \quad T_{i_k}^{k+1} > K_k T^k_{i_k}
    \]
    However, each $\alpha$ term appears exactly twice in these collective inequalities: once in the numerator of \emph{some} $K_a$, and once in the denominator of some $K_b$. This implies that $K_1\times K_2 \times \dots = 1$, and so \emph{at least one $K_i \geq 1$}. This gives us a $T_i$ where $T^{k+1}_i > (1+\epsilon) T_i^k$, and so $T_i$ grows exponentially from $w$. Because $T_i > \mu T_{i-1}/\alpha_i$ for any $i$, each $T_i$ must be bounded below and eventually grow to infinity.
    \end{proof}

    This completes the proof that $c$ is stable under BRD, and hence is stable under RD by Theorem~\ref{thm: RD stability = FP stability}. Lastly, we argue that this implies that this set is an attractor for RD. By Theorem~5.2 of \cite{biggar_replicator_2023}, any attractor of this game contains the `content' of the sink equilibrium (which in this case is exactly the union of the 1-dimensional strategy spaces corresponding to arcs). To show this is itself an attractor, we need only show it is asymptotically stable, which requires that it is stable both on the interior of the space (which we get from Theorem~\ref{thm: RD stability = FP stability} and the stability of $c$ for BRD) and also on the boundary. On the boundary, where the cycle intersects some subgame, it is either still a cycle that is a sink equilibrium (where it is stable by an inductive argument) or it is a best-response path which ends in a PNE within that subgame. 
\end{proof}

\section{Conclusions and open problems} \label{sec: conclusions}

In this paper we studied the phenomenon of stable periodic behavior in game dynamics, and examined how we can identify such behavior computationally. We showed that a sequence of profiles appears as a stable sequence of play from an open set of initial conditions under BRD and RD if and only if it satisfies the eigenvector test derived from the Poincar\'e matrix. Further we showed in Theorem~\ref{thm: sink equilibria are stable} that this test is always satisfied when the cycle in question is a sink equilibrium of the game, connecting this stability analysis with the recent literature on sink equilibria and the structure of the preference graph \cite{goemans_sink_2005,biggar_replicator_2023,biggar_attractor_2024,biggar2025sink}. Several fascinating questions remain open, along the following themes:
\begin{itemize}
    \item (Computing stable cycles:) Our approach gave a computational determination of when a \emph{given} cycle of play is stable, but we did not establish the difficulty of \emph{finding} a stable cycle, or determining if any exist. In other words, we showed that finding stable cycles lies in NP, but a fundamental follow-up question is whether it lies in P, or alternatively it is NP-complete. Even if the general case is intractable, there may be special classes (such as the sink equilibrium case we considered in Section~\ref{sec: cyclic sink equilibria}) which can be efficiently computed, and these would be valuable to identify.

    \item (Structural properties of stable cycles:) Although our Poincar\'e matrix test can determine the stability of a cycle, the cycles that pass that test have many special properties. Determining these properties is important both conceptually---to understand what kinds of cycles can be stable---and computationally, in order to more efficiently search the space of possible cycles. As an example, $2\times n$ games have cycles of many lengths, and yet \emph{the only stable ones are 4-cycles}. We conjecture the following: \emph{in a two-player game, every stable cycle alternates players}. This seems plausible, yet we do not know of the tools to prove it.

    \item (Characterizing more general limit sets:) This paper focused on the special case of cycles, but the broader goal is to characterize---and compute---all attractors of game dynamics \cite{papadimitriou_game_2019,biggar_attractor_2024,biggar2025sink}. Our hope is that our techniques developed here might extend to non-periodic behavior, and so provide a new angle with which to tackle the general problem.
\end{itemize}

\bibliographystyle{plain}
\bibliography{refs}

\end{document}